\documentclass[lettersize,journal]{IEEEtran}
\usepackage{amsmath,amsfonts}
\usepackage{algorithm,algorithmic}
\usepackage{array}
\usepackage[caption=false,font=normalsize,labelfont=sf,textfont=sf]{subfig}
\usepackage{textcomp}
\usepackage{stfloats}
\usepackage{url}
\usepackage{verbatim}
\usepackage{graphicx}
\usepackage{cite}
\usepackage{amssymb}
\usepackage{amsthm}
\usepackage{hyperref}
\usepackage{subcaption}
\usepackage{newtxtext}
\usepackage{newtxmath}
\usepackage{caption}
\newtheorem{theorem}{Theorem}

\begin{document}
\title{Simultaneous Secrecy and Covert Communications (SSACC) in Mobility-Aware RIS-Aided Networks}
\author{
    Yanyu~Cheng,
    Yujian~Hu, 
    Haoran~Liu,
    Hua~Zhong,
    Wei~Wang,
    and~Pan~Li,~\IEEEmembership{Fellow,~IEEE}
    
    \thanks{Yanyu Cheng is with the School of Electrical and Electronic Engineering, Nanyang Technological University, Singapore 639798 (e-mail: ycheng022@e.ntu.edu.sg).
    
    Yujian Hu, Haoran Liu, and Hua Zhong are with the School of Cyberspace, Hangzhou Dianzi University, Hangzhou 310018, China (e-mail: \{255270191; 232270067; hzhong\}@hdu.edu.cn).
    
    Wei Wang is with the School of Information and Communications Engineering, Xi'an Jiaotong University, Xi’an 710049, China (e-mail: w25wang@xjtu.edu.cn).

    Pan Li is with the School of Computer Science, Hangzhou Dianzi University, Hangzhou 310018, China (e-mail: lipan@ieee.org).}
}
\maketitle

\begin{abstract}
In this paper, we propose a simultaneous secrecy and covert communications (SSACC) scheme in a reconfigurable intelligent surface (RIS)-aided network with a cooperative jammer. The scheme enhances communication security by maximizing the secrecy capacity and the detection error probability (DEP). Under a worst-case scenario for covert communications, we consider that the eavesdropper can optimally adjust the detection threshold to minimize the DEP. Accordingly, we derive closed-form expressions for both average minimum DEP (AMDEP) and average secrecy capacity (ASC). To balance AMDEP and ASC, we propose a new performance metric and design an algorithm based on generative diffusion models (GDM) and deep reinforcement learning (DRL). The algorithm maximizes data rates under user mobility while ensuring high AMDEP and ASC by optimizing power allocation. Simulation results demonstrate that the proposed algorithm achieves faster convergence and superior performance compared to conventional deep deterministic policy gradient (DDPG) methods, thereby validating its effectiveness in balancing security and capacity performance.
\end{abstract}


\section{Introduction}
\IEEEPARstart{D}{riven} with breakthroughs in integrated space-air-ground networks, full-spectrum access, ultra-dense heterogeneous systems, and secrecy communications, 6G is positioned to play a pivotal role as a catalyst for economy-wide and society-wide digital transformation. Whether in the military or civilian domain, individuals inevitably utilize wireless channels to transmit critical information. Traditionally, the security of wireless communications relies on key-based encryption (cryptographic)\cite{yang2024from5gto6g}. Although cryptographic can protect transmitted information from interception to some extent, encryption technologies can become vulnerable with advances in high-performance computing technology \cite{Bhat2021}. Consequently, physical layer security (PLS) technologies have emerged to address this need. Unlike conventional security measures applied at higher layers, the security provided by PLS is fundamentally based on the physical properties of the wireless communication channel. 

Traditionally, PLS refers to secrecy communications (SC), and SC draws on the information-theoretic principles of modern cryptography, first established by Shannon in \cite{shannon1949communication}. The concept of secrecy capacity was subsequently established by Wyner \cite{wyner1975wire}, who examined a discrete memoryless wiretap channel and demonstrated the feasibility of secure transmission against an eavesdropper. SC leverages the inherent stochastic properties of noise and communication channels to restrict the amount of information accessible to an illegitimate receiver at the bit level. The theoretical limit of SC is quantified by the secrecy capacity, which is achievable provided that the primary channel quality exceeds that of the eavesdropper \cite{wyner1975wire}. This approach capitalizes on the intrinsic randomness of wireless media, combined with specialized coding schemes and signal processing methods, to ensure confidentiality of communication \cite{8014490,8653906,9170262}. However, SC is constrained by secrecy capacity and cannot guarantee protection for information beyond this capacity. In complex network environment, the channel conditions for an eavesdropper can become comparable or superior to the main channel, leading to a significant erosion of secrecy capacity \cite{OntheSecrecy}. 

To address the aforementioned issue, covert communication (CC), as a novel PLS technology, was proposed \cite{10090449}. The fundamental principle of CC is to ensure covertness by minimizing the probability of detection by adversaries. However, the pursuit of covertness introduces inherent trade-offs, often leading to reduced data rates, limited communication range, and diminished spectral efficiency. Consequently, finding a method that satisfies all security and performance requirements presents significant challenges. Moreover, in dynamic networks, imperfect channel state information (CSI) can degrade the effectiveness of CC \cite{MAKHDOOM2022102784}. The inherently dynamic nature of the wireless environment adds another layer of complexity to the design of reliable CC systems. Therefore, to meet the multifaceted demands of future wireless networks effectively, adaptive security strategies are essential. Such strategies must maintain a holistic balance among covertness, secrecy, and efficiency to guaranty robust security in dynamic scenarios \cite{10210075}.

\subsection{Motivation and Contributions}
SC and CC provide distinct levels of security: SC safeguards the confidentiality of the information content, while CC conceals the very existence of the communication. Balancing these two approaches can effectively minimize the probability of detection while minimizing the risk of information leakage. This paper introduces a framework, simultaneous secrecy and covert communication (SSACC), to realize the balance between security and covert performance. We analyze its performance and optimize the power allocation to enhance overall system security while considering user mobility. The primary contributions are as follows.
\begin{itemize} 
    \item {We propose an SSACC framework for reconfigurable intelligent surface (RIS)-aided systems. In this framework, the RIS extends coverage, while we employ a cooperative jammer to safeguard transmissions between the transmitter and the intended receiver.} 
    \item {Assuming a worst-case scenario for covert communications, we analyze the warden’s optimal detection threshold for minimizing detection error probability (DEP) and derive a closed-form expression for the average minimum DEP (AMDEP).} 
    \item {We introduce a novel performance metric, security quality of experience (QoE), which is a weighted sum of secrecy capacity and effective covert rate protected by the DEP. This metric captures the comprehensive security performance by accounting for information transmission that exceeds the secrecy capacity but remains protected via covert mechanisms. Accordingly, a closed-form expression for the average secrecy capacity is derived for evaluation.}
    \item {To address the dynamic channel conditions induced by user mobility, we propose a deep reinforcement learning (DRL) algorithm leveraging generative diffusion models (GDM) to optimize power allocation between the transmitter and the jammer. This algorithm aims to maximize the QoE in the receiver while ensuring high levels of AMDEP and secrecy capacity.} 
    \item {The effectiveness of the proposed algorithm is validated through simulations that demonstrate the robustness of the developed scheme. Moreover, the GDM-based policy learner achieves faster convergence and superior performance compared to the conventional deep deterministic policy gradient (DDPG) method in the context of SSACC.} 
\end{itemize}

\subsection{Organization}
The remainder of this paper is structured as follows. Section \ref{sec:System Module} details the proposed SSACC system model. Performance metrics and channel statistics are presented in Section \ref{sec:Performance Metrics and Channel Statistics}. In Section \ref{sec:Performance Analysis and QoE Definition}, we provide a comprehensive analysis of the warden's DEP and the legitimate user's secrecy capacity and define the proposed QoE metric. Section \ref{sec:Power Allocation Optimization for Maximizing QoE} proposes a GDM-based DRL algorithm to optimize power allocation for maximizing the QoE. Section \ref{sec:Numerical Result} presents numerical results to validate the theoretical analysis and demonstrate the superiority of the proposed algorithm over benchmarks. Finally, Section \ref{sec:Conclusion} concludes the paper.

\section{Related Work}
\label{sec:Related Work}
This section surveys recent advances in SC and CC, examining foundational models and representative methodologies along with state-of-the-art developments. Finally, it highlights critical challenges and suggests promising avenues for future research.

\subsection{Secrecy Communication}
SC has attracted significant attention as a promising approach to ensuring the security of wireless communications without relying on traditional cryptographic methods. The concept of the wiretap channel was first proposed by Wyner in his seminal work \cite{wyner1975wire}, which demonstrated that information theoretic secrecy is achievable at the physical layer by utilizing the disparity between the main and eavesdropping channels, without relying on encryption keys. Building on this foundational work, the model was further extended by Csiszár and Körner to encompass broadcast channels with confidential information. Their research \cite{1055892} further characterized the secrecy capacity for a broader class of non-degraded wiretap channels.

Building on these foundations, recent studies have focused on secrecy enhancement through resource allocation and optimization techniques. By simultaneously optimizing both the transmit power allocation and the artificial noise (AN) power splitting ratio, Xing et al. achieved a reduction in outage probability for delay-constrained secure communication, employing a dual decomposition and alternating optimization framework as detailed in \cite{7018095}. In \cite{7063622}, the location of an energy-harvesting node was optimized under PLS constraints to minimize the secrecy outage probability. The work in \cite{7841596} proposed a framework to jointly minimize secrecy outage probability and maximize average harvested energy within simultaneous wireless information and power transfer (SWIPT) systems. In \cite{6094170}, the authors addressed power allocation in an MISO system, aiming to minimize total transmit power while meeting a target secrecy probability. A multiuser MISO network with friendly jamming was studied in \cite{7792199}, where optimal power allocation strategies were developed to minimize the total power allocated to both information and jamming signals while satisfying secure quality of experience (QoE) requirements. In \cite{7982788}, a closed-form power allocation solution was derived for non-orthogonal multiple access (NOMA) systems to minimize total transmit power while ensuring secrecy.

\subsection{Covert Communication}
CC, sometimes termed low probability of detection (LPD) communication, has garnered considerable attention from both academic and industrial research communities \cite{10090449}. To achieve undetectable transmission, various strategies have been proposed based on the behavior and capabilities of the warden. In \cite{9838318}, Du et al. developed a power allocation algorithm utilizing particle swarm optimization for UAV networks equipped with multiple antennas and a single warden (Willie). In \cite{chen2021multi}, Chen et al. proposed a scheme designed to obfuscate signal detection by a warden with uncertain location information. Furthermore, Cheng et al. in \cite{10210075} proposed an optimal power allocation scheme for a RIS-aided NOMA network assisted by a friendly jammer, aiming to counter a warden employing a dynamic eavesdropping strategy. The aim was to optimize the connectivity throughput from a multi-antenna transmitter to a full-duplex jamming receiver, subject to a constraint on covert outage probability.

Collaboration among multiple friendly or adversarial nodes has been explored as a critical factor influencing the performance of covert systems. In \cite{jiang2020covert}, it was demonstrated that coordinated warden behavior can significantly reduce DEP, thereby posing a serious threat to covert communication reliability. Soltani et al. in \cite{8445707} proposed a strategy wherein the friendly node closest to the warden generates artificial noise, considering both single and multiple collaborating warden scenarios. The investigation of multi-hop covert communication was conducted in \cite{8315147}, with a focus on networks featuring multiple collaborating wardens. In particular, most existing studies assume static locations of the warden, which can underestimate their detection capacity. In practical scenarios, wardens can adaptively change their positions and strategies to improve detection performance. To address this, \cite{9099441} proposed an iterative algorithm employing a multi-antenna jammer to mitigate the impact of Alice’s imperfect knowledge of the warden’s location.

\subsection{Joint SC and CC}
Recently, the integration of physical layer security and covert communications has attracted increasing attention, particularly with the aid of advanced hardware architectures like simultaneously transmitting and reflecting (STAR)-RIS and movable antennas (MA). Unlike traditional approaches that focus solely on either secrecy or covertness, these emerging works explore the simultaneous satisfaction of distinct security requirements for multiple users. For instance, Hu et al. \cite{3354452} proposed a STAR-RIS enhanced framework to serve a covert user and a security user simultaneously. By optimizing the transmission and reflection coefficients, they achieved low detection probability for the covert user while maximizing the secrecy capacity for the security user. In a similar vein, Zhang et al. \cite{10333375} utilized STAR-RIS to assist NOMA communications, guaranteeing the covertness of the weak user while maintaining the connection for the strong user. Similarly, Wu et al. \cite{Wu_Entropy2025} investigated MA-aided STAR-RIS systems, leveraging the spatial degrees of freedom to enhance secrecy rates in integrated sensing and communication scenarios. However, it is worth noting that these works primarily achieve joint performance by spatially separating resources to serve different users with specific security needs (i.e., strengthening covertness for one user and secrecy for another separately), rather than unifying them into a comprehensive metric for a single transmission link.
\section{System Model}
\label{sec:System Module}
The SSACC system model under consideration, illustrated in Fig. \ref{System Model}, includes a controlled friendly jammer (Jammer), a legitimate transmitter (Alice), a receiver (Bob), and a reconfigurable intelligent surface (RIS) \cite{9108996,9424177}. A warden (Willie) attempts to detect and decode the communications exchanged between Alice and Bob. Specifically, located in a non line of sight region without a direct link to Alice, Bob relies on the RIS for communication support. The jammer, which is operated by Alice, is utilized to disrupt Willie's ability to detect the data exchange. 

The RIS comprises $N$ reflecting elements, with the complex gain of the $n$-th element given by $\beta_{n}e^{j\theta_{n}}$ $(j=\sqrt{-1})$, where $\beta_n \in {[0,1]}$ represents the amplitude coefficient, and $\theta_n \in {[0,2\pi)}$ denotes the phase shift coefficient. These parameters control the amplitude (or attenuation due to passive reflection) and phase shift of the reflected signal. The complex gain matrix of the RIS is expressed as $\mathbf{\Theta} = diag(\beta_{1}e^{j\theta_{1}},\beta_{2}e^{j\theta_{2}},...,\beta_{N}e^{j\theta_{N}})$ $(j=\sqrt{-1})$. Signals undergoing multiple reflections by the RIS are neglected due to substantial path loss \cite{wu2019towards}.
\begin{figure}[t]
    \centering
    \includegraphics[trim=40 20 10 20,clip,draft=false,width=0.5\textwidth]{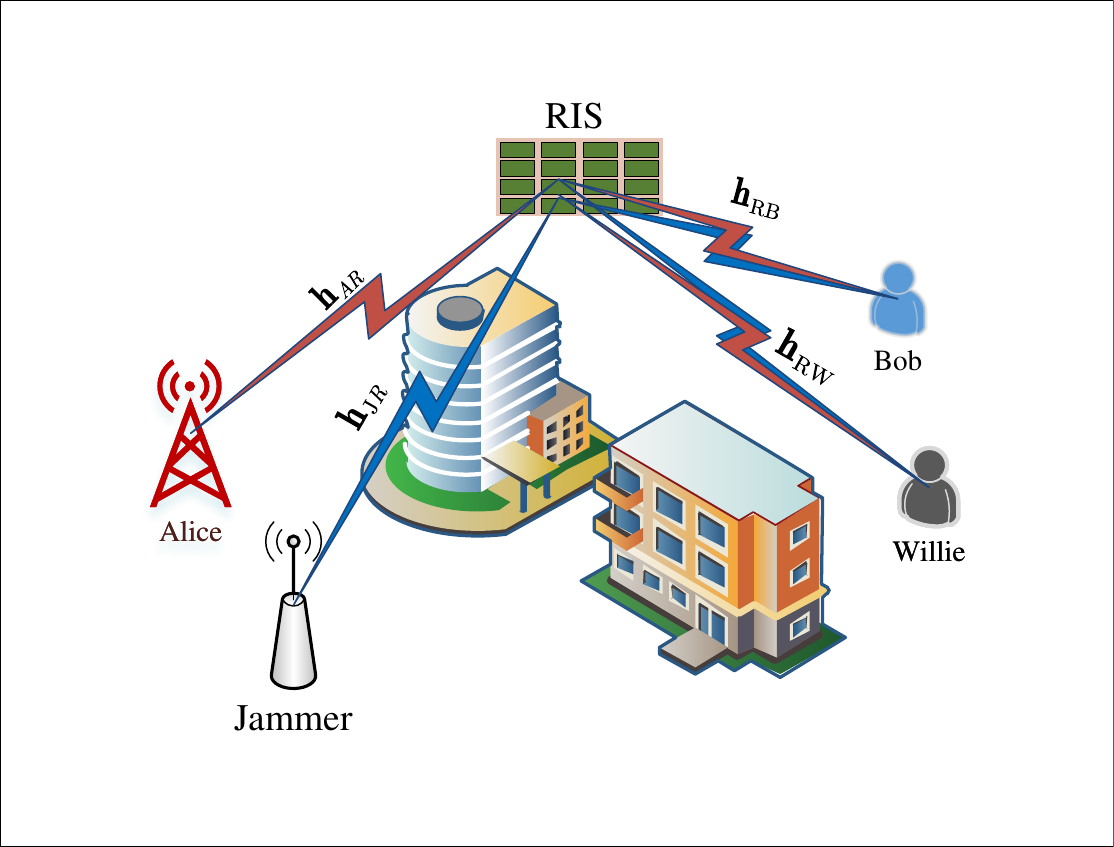}
    \caption{System model for SSACC system. The figure depicts a scenario with only reflection links, where Alice is the legitimate transmitter, Jammer is the friendly jammer, Bob is the legitimate receiver, and Willie is the eavesdropper.} \label{System Model}
\end{figure}

\subsection{Channel Model}
This paper employs a quasi-static flat fading channel model to characterize the wireless communication. The channel response vectors between Alice/Jammer and the RIS are represented by $\mathbf{h}_{AR}$ and $\mathbf{h}_{JR}$, respectively, while the fading coefficient vectors from the RIS to Bob and Willie are represented by $\mathbf{h}_{RB}$ and $\mathbf{h}_{RW}$. Each of these vectors has a dimension of $N \times 1$. All fading coefficients are assumed to follow a $\text{Nakagami-}m$ fading model with fading parameter $m_{\chi}$, where $\chi \in\{AR, JR, RB, RW\}$ corresponds to the respective fading links \cite{cheng2021downlink}. In addition, each channel experiences path loss, characterized by a path loss coefficient $\alpha_{\chi}$.

Regarding CSI availability, the following assumptions are made:
\begin{itemize} 
    \item{Alice has access to the instantaneous CSI of the Alice-RIS-Bob and Jammer-RIS-Bob links, as the Jammer is under Alice’s control \cite{wu2019towards}.} 
    \item{Alice only possesses statistical CSI for the RIS-Willie links, since Willie aims to conceal its presence from the legitimate system, making real-time CSI acquisition challenging for Alice \cite{lv2021covert}.} 
    \item{Willie is assumed to have perfect knowledge of CSI for all links, representing the worst-case scenario for covert communication \cite{lv2021covert}.} 
\end{itemize}
Specifically, various channel estimation techniques have been proposed for RIS-enabled systems to acquire accurate CSI \cite{9366805,10210075}.

\subsection{Signal Model}
Alice transmits a signal sequence $x_A(k) = \sqrt{P_A}s(k), k=1,2,\dots,K$, where $P_A$ represents Alice’s transmit power and $s(k)$ $(\mathbf{E}(|s(k)|^2) = 1)$ is the message intended for Bob. Additionally, the Jammer sends an interference signal sequence $x_J(k) = \sqrt{P_J}s_J(k), k=1,2,\dots,K$, where the jamming power $P_J$ is a random variable uniformly distributed over $[0, \hat{P}_J]$, adding a level of unpredictability and uncertainty \cite{9367492}. Its probability density function (PDF) can be expressed as
\begin{equation}\label{PDF of the Jammer power}
    f_{P_J}\left( x \right)=\left\{\begin{array}{l}
        \frac{1}{\hat{P}_J},\ 0\le x\le \hat{P}_J,\\
        0,\ \ \ \text{otherwise}.\\
    \end{array} \right. 
\end{equation}
The signal obtained by Bob is given by:
\begin{equation}\label{The signal of Bob}
    \begin{split}
    y_B\left( x \right) &=\mathbf{h}_{RB}^{T}\mathbf{\Theta h}_{AR}\sqrt{\mathcal{L}_1} x_A\left( k \right) \\
    &+ \mathbf{h}_{RB}^{T}\mathbf{\Theta h}_{JR}\sqrt{\mathcal{L}_2}x_J\left( k \right) +n_B\left( k \right) ,
    \end{split}
\end{equation}
where $\sqrt{\mathcal{L}_1}$ and $\sqrt{\mathcal{L}_2}$ are the path loss of Alice-RIS-Bob and Jammer-RIS-Bob links, respectively. Here, $\mathcal{L}_1=L(d_{RB})L(d_{AR})$, $\mathcal{L}_2=L(d_{RB})L(d_{JR})$, where $L(d_\chi)=d_{\chi}^{-\alpha _{\chi}}$, $d_\chi$ is the distance, $a_\chi$ is the path loss coefficient, and ${\chi}\in\{AR,JR,RB,RW\}$\cite{Wang_COMST2023}.

Bob receives the signal with a signal-to-interference-plus-noise ratio (SINR) that is formulated as
\begin{equation}\label{SINR B}
    \psi _B=\frac{P_A|\mathbf{h}_{RB}^{T}\mathbf{\Theta h}_{AR}|^2\mathcal{L}_1}{\rho P_J|\mathbf{h}_{RB}^{T}\mathbf{\Theta h}_{JR}|^2\mathcal{L}_2+\sigma _{B}^{2}},
\end{equation}
where ${\rho} \in \left[ 0,1 \right]$ is the self-interference cancelation coefficient.

Willie aims to identify the presence of communication and recover the information exchanged between Alice and Bob. This leads to a binary detection problem. Specifically, the signals observed by Willie correspond to one of two distinct hypotheses: the null hypothesis $\mathcal{H}_0$, which represents the case where Alice transmits no signal to Bob, and the alternative hypothesis $\mathcal{H}_1$, which corresponds to an active transmission from Alice to Bob. The expressions for the received signals under each hypothesis are given as follows:
\begin{equation}\label{H0}
    \begin{split}
    \mathcal{H}_0:y_W\left( x \right) =\mathbf{h}_{RW}^{T}\mathbf{\Theta h}_{JR}\sqrt{\mathcal{L}_3} x_J\left( k \right)+n_W\left( k \right),
    \end{split}
\end{equation}
\begin{equation}\label{H1}
\begin{split}
\mathcal{H}_1:y_W\left( x \right)=&\mathbf{h}_{RW}^{T}\mathbf{\Theta h}_{JR}\sqrt{\mathcal{L}_3} x_J\left( k \right) \\
&+\mathbf{h}_{RW}^{T}\mathbf{\Theta h}_{AR}\sqrt{\mathcal{L}_4} x_A\left( k \right) +n_W\left( k \right), 
\end{split}
\end{equation}
where $\mathcal{L}_3=L(d_{RW})L(d_{JR})$ and $\mathcal{L}_4=L(d_{RW})L(d_{AR})$. 

Willie performs binary detection using a radiometer based on (\ref{H0}) and (\ref{H1}). Observing the average power received at Willie, $P_W=\frac{1}{K}\sum_{k=1}^K{|y_W\left( k \right) |^2}$, Willie’s decision rule takes the following rule
\begin{equation}\label{PW_Detection}
P_W\underset{\mathcal{D}_0}{\overset{\mathcal{D}_1}{\gtrless}}\tau.
\end{equation}
Here, $\tau > 0$ represents Willie’s detection threshold, while $D_0$ and $D_1$ indicate binary decisions under hypotheses $\mathcal{H}_0$ and $\mathcal{H}_1$, respectively \cite{10210075}. We assume that Willie has access to an unlimited quantity of signal samples to perform binary hypothesis testing, i.e., $K \to \infty$ \cite{lv2021covert}. Consequently, the uncertainties associated with the transmitted signals and received noise are eliminated, and the average received power at Willie is expressed as
\begin{equation}\label{PW_2}
    P_W=\left\{ \begin{array}{l}
        \zeta _1P_J+\zeta _2,\ \ \ \mathcal{H}_0,\\
        \zeta _1P_J+\zeta _3,\ \ \ \mathcal{H}_1,\\
    \end{array} \right.
\end{equation}
where $\zeta_1=\mathcal{L}_3|\mathbf{h}_{RW}^{T}\mathbf{\Theta h}_{JR}|^2$, $\zeta_2=\sigma_{W}^2$, and $\zeta_3=\mathcal{L}_4P_A|\mathbf{h}_{RW}^{T}\mathbf{\Theta h}_{AR}|^2+\sigma_{W}^2$.

To achieve SC, Alice must ensure that her transmission rate does not exceed the secrecy capacity, which is defined as
\begin{equation}\label{R_S}
R_S=[\log_2(1+\psi_B)-\log_2(1+\psi_W)]^+.
\end{equation}
Based on (\ref{H1}), $\psi_W$, representing the SINR of Willie’s received signal under the alternate hypothesis, is expressed as
\begin{equation}\label{SINR W}
    \psi _W=\frac{P_A|\mathbf{h}_{RW}^{T}\mathbf{\Theta h}_{AR}|^2\mathcal{L}_4}{ P_J|\mathbf{h}_{RW}^{T}\mathbf{\Theta h}_{JR}|^2\mathcal{L}_3+\sigma _{W}^{2}}.
\end{equation}

\section{Performance Metrics and Channel Statistics}
\label{sec:Performance Metrics and Channel Statistics}
In this section, we commence by defining the performance metrics employed in our analysis, followed by an introduction of the CSI required for performance evaluation. This paper focuses primarily on two key system metrics: secrecy capacity and DEP.

\subsection{Performance Metrics}
\subsubsection{Secrecy Capacity} Let $R_B$ and $R_W$ represent the Shannon capacities of the main channel and the eavesdropper channel, respectively\cite{Akbar_OJCOMS2024}. The secrecy capacity $R_S$ is given by:
\begin{equation}\label{Shannon secrecy capacity}
    R_S=[R_B-R_W]^+=[\log_2(1+\psi_B)-\log_2(1+\psi_W)]^+,
\end{equation}
where $R_B = \log_2(1+\psi_B)$ and $R_W=\log_2(1+\psi_W)$.

\subsubsection{Detection Error Probability} Willie’s detection error occurs when he makes an incorrect decision on binary detection,  either a false alarm, where he detects a transmission during Alice's silence, or a miss detection, where he fails to identify an ongoing transmission. Its formal definition is given by \cite{tao2021achieving,li2020covert}
\begin{equation}\label{DEP}
    \mathbb{P}_{DE}=\mathbb{P}_{MD}+\mathbb{P}_{FA}, 
\end{equation}
where $ \mathbb{P}_{MD}=P_{r}\{\mathcal{D}_0|\mathcal{H}_1\} $ denotes the miss detection probability (MDP), $ \mathbb{P}_{FA}=P_{r}\{\mathcal{D}_1|\mathcal{H}_0\} $ denotes the false alarm probability (FAP), and $\mathbb{P}_{DE} \in [0,1] $. 

\subsection{Channel Statistics}
To analyze system performance, an accurate CSI is required.
\subsubsection{Alice-RIS-Bob Link} The reflective properties of RIS can be intelligently controlled to provide Bob with optimal channel quality. The channel gain for the Alice–RIS–Bob link is expressed as
\begin{equation}\label{channel gain of ARB}
    \left\lvert\mathbf{h}_{RB}^T\mathbf{\Theta}\mathbf{h}_{AR}\right\rvert=\left\lvert \sum_{n = 1}^{N}\beta_nh_{RB,n}h_{AR,n}e^{j\theta_n}\right\rvert,
\end{equation}
where $h_{RB,n}$ and $h_{AR,n}$ are the $n$th element of $\mathbf{h}_{RB}$ and $\mathbf{h}_{AR}$, respectively. To maximize channel gain, phase shifts are aligned with co-phase incoming and outgoing signals\cite{8811733}. Therefore, the channel gain under the optimal $\{\theta^*_n\}$ is 
\begin{equation}\label{optimal channel gain}
    \left\lvert\mathbf{h}_{RB}^T\mathbf{\Theta}\mathbf{h}_{AR}\right\rvert=\beta^2\left(\sum_{n = 1}^{N}|h_{RB,n}||h_{AR,n}|\right)^2,
\end{equation}
where $\beta_n=\beta,\forall_n$ without loss of generality. The PDF characterizing the equivalent channel gain under the optimal phase shift configuration ${\theta^*_n}$ can be formulated as follows: $|g_{ARB}|^2=\frac{\left( \sum\limits_{n=1}^N{|h_{RB,n}||h_{AR,n}|} \right) ^2}{N(1-\mu)}$, where $\mu =\frac{1}{m_{RB}m_{AR}}\left( \frac{\Gamma \left( m_{RB}+\frac{1}{2} \right)}{\Gamma \left( m_{RB} \right)} \right) ^2\left( \frac{\Gamma \left( m_{AR}+\frac{1}{2} \right)}{\Gamma \left( m_{AR} \right)} \right) ^2$. As the number of reflecting elements $N$ becomes large, $|g_{ARB}|^2$ asymptotically follows a noncentral chi-square distribution, denoted as $|g_{ARB}|^2\overset{\raisebox{-0.2ex}{.}}{\thicksim }\chi^{\prime 2}_1(\lambda)$, where $\lambda =\frac{N\mu}{1-\mu}$. Its PDF and CDF have been given by \cite{10210075}
\begin{equation}\label{PDF of ARB}
	f_{|g_{ARB}|^2}\left( x \right) =e^{-\frac{x+\lambda}{2}}\sum_{i=0}^{\infty}{\frac{\lambda ^ix^{i-\frac{1}{2}}}{i!2^{\left( 2i+\frac{1}{2} \right)}\Gamma \left( i+\frac{1}{2} \right)}},
\end{equation}
and
\begin{equation}\label{CDF of ARB}
	F_{|g_{ARB}|^2}\left( x \right) =e^{-\frac{\lambda}{2}}\sum_{i=0}^{\infty}{\frac{\lambda ^i\gamma \left( i+\frac{1}{2},\frac{x}{2} \right)}{i!2^i\Gamma \left( i+\frac{1}{2} \right)}},
\end{equation}
for $x \geq 0$, respectively.

\subsubsection{Alice-RIS-Willie, Jammer-RIS-Bob, and Jammer-RIS-Wille Links}For these channels, the phase shifts of the RIS can be treated as random variables, since the RIS is optimized to maximize channel quality at Bob \cite{lv2021covert,9079918}. Denoted as $g_{ARW}=\frac{\mathbf{h}_{RW}^{T}\mathbf{\Theta h}_{AR}}{\beta}=\sum\limits_{n=1}^N{h_{RW,n}h_{AR,n}e^{j\theta _{n}^{*}}}$, $g_{JRW}=\frac{\mathbf{h}_{RW}^{T}\mathbf{\Theta h}_{JR}}{\beta}=\sum\limits_{n=1}^N{h_{RW,n}h_{JR,n}e^{j\theta _{n}^{*}}}$, and $g_{JRB}=\frac{\mathbf{h}_{RB}^{T}\mathbf{\Theta h}_{JR}}{\beta}=\sum\limits_{n=1}^N{h_{RB,n}h_{JR,n}e^{j\theta _{n}^{*}}}$. When $N$ is sufficiently large and $\{\theta_n\}$ is random  \cite{9079918}, $g_{\phi } (\phi \in \{ARW,JRB,JRW\})$ tends to follow a Gaussian distribution with zero mean and  variance $N$. Consequently, the PDF and CDF of $|g_{\phi}|^2 (\phi \in \{ARW,JRB,JRW\})$ are given by
\begin{equation}\label{PDF of E}
    f_{|g_{\phi}|^2}\left( x \right) =\frac{1}{N}e^{-\frac{x}{N}},
\end{equation}
and
\begin{equation}\label{CDF of E}
    F_{|g_{\phi}|^2}\left( x \right) =1-e^{-\frac{x}{N}},
\end{equation}
for $x\geq0$, respectively.

\section{Performance Analysis and QoE Definition}
\label{sec:Performance Analysis and QoE Definition}
This section provides a comprehensive analysis of secrecy capacity and DEP performance. Specifically, we evaluate DEP under both a fixed detection threshold and an optimal threshold in Willie’s binary detection process. This analysis aims to provide information on the effects of threshold selection on the effectiveness of covert communication.

\subsection{Average Minimum Detection Error Probability}

\subsubsection{Fixed Detection Threshold} Based on (\ref{PW_Detection}) and (\ref{PW_2}), the MDP and the false alarm (FA) probability have been given in \cite{10210075}, which are as follows: 
\begin{equation}\label{FAP}
\mathbb{P}_{FA}=\text{P}_{\text{r}}\left( \zeta _1P_J+\zeta _2>\tau \right),
\end{equation}
and
\begin{equation}\label{MDP}
    \mathbb{P}_{MD}=\text{P}_{\text{r}}\left( \zeta _1P_J+\zeta _3\leq \tau \right),
\end{equation}
respectively. $\mathbb{P}_{FA}$ and $\mathbb{P}_{MD}$ under a fixed threshold can be derived as 
\begin{equation}\label{FAP_Fixed}
    \begin{split}
    \mathbb{P}_{FA}&=\text{P}_{\text{r}}\left( P_J>\frac{\tau -\zeta _2}{\zeta _1} \right)  \\
    &=\begin{cases}
        \begin{array}{r r}
            1,                                         &\tau <\tau _1,            \\
            1-\frac{\tau -\zeta _2}{\zeta _1\hat{P}_J},&\tau _1\leq \tau <\tau _2,\\
            0,                                         &\tau \geq \tau _2,        \\
        \end{array}
    \end{cases}
\end{split}
\end{equation}
and
\begin{equation}\label{MDP_Fixed}
    \begin{split}
        \mathbb{P}_{MD}&=\text{P}_{\text{r}}\left( P_J\leq \frac{\tau -\zeta _3}{\zeta _1} \right) \\
    &=\begin{cases}
        \begin{array}{r r}
            0,                                         &\tau <\tau _3,            \\
            \frac{\tau -\zeta _3}{\zeta _1\hat{P}_J},  &\tau _3\leq \tau <\tau _4,\\
            1,                                         &\tau \geq \tau _4,        \\
        \end{array}
    \end{cases}
\end{split}
\end{equation}
where $\tau_1=\zeta_2$, $\tau_2=\zeta_1\hat{P}_J+\zeta_2$, $\tau_3=\zeta_3$, and $\tau_4=\zeta_1\hat{P}_J+\zeta_3$.

Based on (\ref{FAP_Fixed}) and (\ref{MDP_Fixed}), we can derive that the DEP under fixed threshold has two cases as follows: 

$Case \ 1:$ When $\tau_2 \leq \tau_3$, the DEP is given by
\begin{equation}\label{DEP_S1}
    \xi =\begin{cases}
        \begin{array}{l r}
        1,                     &\tau <\tau _1,      \\
        1-\frac{\tau -\zeta _2}{\zeta _1\hat{P}_J},&\tau _1\leq \tau <\tau _2,\\
        0,                     &\tau _2\leq \tau <\tau _3,\\
        \frac{\tau -\zeta _3}{\zeta _1-\hat{P}_J}, &\tau _3\leq \tau <\tau _4,\\
        1,                     &\tau \geq \tau _4.    \\
        \end{array}
    \end{cases}
\end{equation}

$Case \ 2:$ When $\tau_2 > \tau_3$, the DEP is given by
\begin{equation}\label{DEP_S2}
    \xi =\begin{cases}
        \begin{array}{l r}
        1,                     &\tau <\tau _1,      \\
        1-\frac{\tau -\zeta _2}{\zeta _1\hat{P}_J},&\tau _1\leq \tau <\tau _3,\\
        1-\frac{\zeta_3-\zeta_2}{\zeta_1\hat{P}_J},                     &\tau _3\leq \tau <\tau _2,\\
        \frac{\tau -\zeta _3}{\zeta _1-\hat{P}_J}, &\tau _2\leq \tau <\tau _4,\\
        1,                     &\tau \geq \tau _4.    \\
        \end{array}
    \end{cases}
\end{equation}

\subsubsection{Optimal Detection Threshold} It is important to note that the DEP discussed previously is evaluated under a fixed detection threshold. However, in the worst-case scenario for CC, Willie, acting as an intelligent warden, has the ability to optimize the detection threshold to minimize the DEP. Consequently, it is crucial to analyze performance under this worst-case condition.

In the first case $(\tau_2 \leq \tau_3)$, the minimum DEP is clearly 0, which corresponds to $\tau^*\in[\tau_2,\tau_3)$.

For the second case $(\tau_2 > \tau_3)$, the DEP is 1 for $\tau\in[0,\tau_1)\bigcup[\tau_4,\infty]$. The DEP decreases monotonically in $\tau\in[\tau_1,\tau_2)$ and increases in $\tau\in[\tau_3,\tau_4)$. Hence, for any $\tau\in[\tau_2,\tau_3)$, the DEP is the minimum value. Therefore, the optimal threshold set by Willie is given by 
\begin{equation}\label{optimal threshold}
    \tau ^*\in\begin{cases}
        \begin{array}{l r}
        \left[ \tau _2,\tau _3 \right] , &\tau _2<\tau _3,\\
        \left[ \tau _3,\tau _2 \right] ,&    \tau _3<\tau _2.\\
        \end{array}
    \end{cases}
\end{equation}

Based on (\ref{DEP_S2}) and (\ref{optimal threshold}), the DEP of Willie under the optimal detection threshold can be expressed as follows:
\begin{equation}\label{DEP_optimal}
    \mathbb{P}_{MDE}=\begin{cases}
    \begin{array}{l r}
        0,&\tau _2 \leq \tau _3,\\
        1-\frac{\zeta _3 -\zeta _2}{\zeta _1\hat{P}_J},&   \tau _2> \tau _3.\\
        \end{array}
    \end{cases}
\end{equation}
\subsubsection{AMDEP Under the Optimal Detection Threshold} The AMDEP can be derived from the statistical CSI of all channels. Based on (\ref{DEP_optimal}), we note that $\mathbb{P}_{MDE}$ is a piecewise function, allowing us to conclude that
\begin{equation}\label{AMDEP}
    \mathbf{E}(\mathbb{P}_{MDE})=\int_0^{\infty}{\int_0^{\frac{\eta _1y}{\eta _2}}{\mathbb{P}_{MDE}f_{|g_{ARW}|^2}\left( x \right) dxf_{|g_{JRW}|^2}\left( y \right) dy}}.
\end{equation}
\begin{theorem}
Based on (\ref{PDF of E}) and (\ref{AMDEP}), the AMDEP of Willie is given by
\begin{equation}\label{AMDEP Result}
    \begin{split}
    \mathbf{E}(\mathbb{P}_{MDE})=&\int_0^{\infty}{\int_0^{\frac{\eta _1y}{\eta _2}}{\frac{e^{-\frac{x+y}{N}}}{N^2}\left( 1-\frac{\eta _2x}{\eta _1y} \right)}}dxdy \\
    =&1-\frac{\eta _2}{\eta _1}\log \left( \frac{\eta _2+\eta _1}{\eta _2} \right) ,
    \end{split}
\end{equation}
where $\tau_2 > \tau_3$, $\eta _1=P_J\mathcal{L}_3\beta ^2$, and $\eta _2=P_A\mathcal{L}_4\beta ^2$.
\end{theorem}

\subsection{Average Secrecy Capacity}
Secrecy capacity refers to the maximum achievable rate for reliable and secure communication between an authorized transmitter and receiver, while ensuring that the information leakage to any eavesdropper remains constrained. Specifically, the perfect secrecy capacity is defined as the difference between the channel capacity of the legitimate user and that of the eavesdropper \cite{wyner1975wire,5605343}. Based on (\ref{R_S}), the average secrecy capacity (ASC) can be expressed as follows:
\begin{equation}
        \mathbf{E}\left( R_S \right) = \mathbf{E}\left( [R_B - R_W]^+ \right).
\end{equation}

Given that the RIS is optimized to maximize the channel gain for Bob and that the gains of the respective channels are independent, the probability that $\psi_B < \psi_W$ is negligible. Consequently, the ASC can be expressed as

\begin{equation}\label{ASC}
    \mathbf{E}\left( R_S \right) \simeq \mathbf{E}\left( R_B \right) -\mathbf{E}\left( R_W \right), 
\end{equation}
where $\mathbf{E}\left( R_B \right)$ is the average channel capacity of Bob and $\mathbf{E}\left( R_W \right)$ is the average channel capacity of Willie. Hence, $\mathbf{E}\left( R_B \right)$ and $\mathbf{E}\left( R_W \right)$ are required, 
which are obtained in the following theorem after some mathematical manipulations.
\begin{theorem}
$\mathbf{E}\left( R_B \right)$ is given by
\begin{equation}\label{Average RB}
    \mathbf{E}\left( R_B \right) \simeq \sum_{l=1}^{u_1}{\omega _{1,l}\Phi_1\left( x_{1,l} \right)},
\end{equation}
where $u_1$ is the order of the Gauss-Laguerre polynomial, $x_{1,l}$ is the $l$-th root of Laguerre polynomial $L_{u_1}(x)$, $\omega_{1,l}=\frac{x_{1,l}}{(u_1 + 1)^2[L_{u_1 + 1}(x_{1,l})]^2}$ and 
\begin{equation}
    \Phi_1=e^{\frac{x-\lambda}{2}}\sum_{i=0}^{\infty}{\frac{\lambda ^ix^{i-\frac{1}{2}}}{i!2^{\left( 2i+\frac{1}{2} \right)}\Gamma \left( i+\frac{1}{2} \right)}\log _2\left( 1+\nu _1 \right)},
\end{equation}
where $\nu_1= \frac{P_A\mathcal{L}_1N(1-\mu)\beta^2x}{\sigma _{B}^{2}}$. 

$\mathbf{E}(R_W)$ is given by
\begin{equation}\label{Average RW}
    \mathbf{E}\left( R_W \right) \simeq \nu _2\sum_{l=1}^{u_2}{\omega _{2,l}\Phi_2\left( x_{2,l} \right)},
\end{equation}
where $\nu _2=\frac{P_A\mathcal{L}_4}{\hat{P}_J\mathcal{L}_3\ln 2}$, $u_2$ is the order of the Gauss-Laguerre polynomial, $x_{2,l}$ is the $l$-th root of the Laguerre polynomial $L_{u_2}(x)$, $\omega_{2,l}=\frac{x_{2,l}}{(u_2 + 1)^2[L_{u_2 + 1}(x_{2,l})]^2}$ and 
\begin{equation}
    \Phi_2=\frac{e^{\frac{N-y}{N}}}{y}\left[ \ln \left( 1+\nu _3 \right) +e^{\nu_4}\Gamma (\nu_4) -e^{\nu_5}\Gamma(\nu_5)\right],
\end{equation}
where $\nu _3=\frac{\beta ^2\hat{P}_J\mathcal{L}_3y}{\sigma _{W}^{2}}$, $\nu _4=\frac{\sigma _{W}^{2}}{\beta ^2NP_A\mathcal{L}_4}\left( 1+v_3 \right) $, and $\nu _5=\frac{\sigma _{W}^{2}}{\beta ^2NP_A\mathcal{L}_4}$.

Based on (\ref{Average RB}) and (\ref{Average RW}), $\mathbf{E}\left( R_S \right)$ can be approximated as
\begin{equation} \label{Result of ASC}
    \mathbf{E}\left( R_S \right) \simeq \sum_{l=1}^{u_1}{\omega _{1,l}\Phi_1\left( x_{1,l} \right)} - \nu _2\sum_{l=1}^{u_2}{\omega _{2,l}\Phi_2\left( x_{2,l} \right)}.
\end{equation}
\end{theorem}

\begin{proof}
    We first evaluate the average channel capacity of Bob. In this scenario, since the Jammer operates cooperatively under Alice's control, we assume that the interference affecting Bob is completely mitigated, resulting in $\rho=0$. Accordingly, the average channel capacity of Bob can be derived as follows:
    \begin{equation}
        \mathbf{E}\left( R_B \right) =\mathbf{E}\left( \log _2\left( 1+\varpi _1|g_{ARB}|^2 \right) \right),
    \end{equation} 
    where $\varpi _1 = \frac{P_A\mathcal{L}_1N(1-\mu)\beta^2}{\sigma _{B}^{2}}$. Based on (\ref{PDF of ARB}), the average channel capacity of Bob can be expressed as 
    \begin{equation}
        \mathbf{E}\left( R_B \right) =\underset{\mathcal{J}_1}{\underbrace{\int_0^{\infty}{e^{-\frac{x+\lambda}{2}}\sum_{i=0}^{\infty}{\frac{\lambda ^ix^{i-\frac{1}{2}}}{i!2^{\left( 2i+\frac{1}{2} \right)}\Gamma \left( i+\frac{1}{2} \right)}\log _2\left( 1+\varpi _1 x\right)}d_x}}}.
    \end{equation}
    Using the Gauss–Laguerre quadrature, we approximate $\mathcal{J}_1$ by
    \begin{equation}
        \mathcal{J}_1\backsimeq \sum_{l=1}^{u_1}{\omega _{1,l}\Phi_1\left( x_{1,l} \right)}.
    \end{equation}
    
    The expectation value of $R_W$ is given by (\ref{RW Expression}) shown at the bottom of the next page, where $\varpi_2=\beta^2P_A\mathcal{L}_4$ and $\varpi_3=\beta^2\mathcal{L}_3$. Based on (\ref{PDF of the Jammer power}), since $P_J$ follows a uniform distribution, it can be easily derived that   
    \begin{figure*}[hb]
        \centering   
        \hrule 
        \begin{equation}\label{RW Expression}
            \begin{split}
            \mathbf{E}\left( R_W \right) =&\mathbf{E}\left( \log _2\left( 1+\frac{\varpi _2|g_{ARW}|^2}{\varpi _3P_J|g_{JRW}|^2+\sigma _{W}^{2}} \right) \right) \\
            =&\underset{\mathcal{J}_4}{\underbrace{\int_0^{\infty}{\underset{\mathcal{J}_3}{\underbrace{\int_0^{\infty}{e^{-\frac{x+y}{N}}\underset{\mathcal{J}_2}{\underbrace{\int_0^{\hat{P}_J}{\log _2\left( 1+\frac{\varpi _2x}{\varpi _3yz+\sigma _{W}^{2}} \right) f_{P_J}\left( z \right) dz}}f_{|g_{ARW}|^2}\left( x \right)}dx}}f_{|g_{JRW}|^2}}}\left(y\right)dy}}.
        \end{split}
        \end{equation}
    \end{figure*}
    \begin{equation}\label{First JF}
        \mathcal{J}_2= \frac{\nu _6}{y}\left[ \nu_7\ln(\nu_7) - \nu_8\ln(\nu_8) - \nu_9\ln(\nu_9)+ \nu_{10}\ln(\nu_{10})\right],
    \end{equation}
    where $\nu_6=\frac{1}{\hat{P}_J\mathcal{L}_3\ln 2}$, $\nu_7=\frac{\sigma _{W}^{2}}{\beta ^2}$, $\nu_8=P_A\mathcal{L}_4x+\frac{\sigma _{W}^{2}}{\beta ^2}$, $\nu_9=\frac{\sigma _{W}^{2}}{\beta ^2}+\hat{P}_J\mathcal{L}_3y$, and $\nu_{10}=P_A\mathcal{L}_4x+\frac{\sigma _{W}^{2}}{\beta ^2}+\hat{P}_J\mathcal{L}_3y$.
    Based on (\ref{PDF of E}),(\ref{RW Expression}) and (\ref{First JF}), $\mathcal{J}_3$ can be derived from the channel gain distribution of ARW, which is given by
    \begin{equation}\label{Second JF}
        \mathcal{J}_3=\frac{\nu_{2}}{y}\left[ \ln \left( \frac{\nu _{11}}{\sigma _{W}^{2}} \right) +g\left( \frac{\nu _{11}}{N^2\varpi _2} \right) -g\left( \frac{\sigma _{W}^{2}}{N^2\varpi _2} \right) \right],
    \end{equation}
    where $\nu_{11}=\sigma _{W}^{2}+\beta ^2\hat{P}_JL_3y$ and $g(\cdot)$ is a function which is given by 
    \begin{equation}
        g(x)=e^{x}\Gamma(0,x).
    \end{equation}
    $\mathcal{J}_4$ can be expressed as
    \begin{equation}
        \mathcal{J}_4=\int_0^{\infty}{e^{-\frac{y}{N}}}\mathcal{J}_3dy.
    \end{equation}
    Next, $\mathcal{J}_4$ can be approximated using Gauss-Laguerre quadrature as $\mathcal{J}_4\backsimeq  \nu _2\sum_{l=1}^{u_2}{\omega _{2,l}\Phi_2\left( x_{2,l} \right)}$. Finally, we can derive $\mathbf{E}(R_W)$ as in (\ref{Average RW}). This completes the proof.
\end{proof}

\subsection{Quality of Experience}
We introduce a novel performance metric that integrates secrecy capacity with DEP. Willie can obtain information beyond secrecy capacity, where Alice is transmitting, but because of the existence of the DEP, unsafe information is protected by CC. The proposed QoE is given by
\begin{equation}\label{QoE}
    QoE = \alpha \mathbf{E}(R_S) + \varsigma \mathbf{E}(R_W) \mathbf{E}(\mathbb{P}_{MDE}),
\end{equation}
where $\alpha$ and $\varsigma$ represent the weights given to SC and CC, respectively.

\section{Power Allocation Optimization for Maximizing QoE}
\label{sec:Power Allocation Optimization for Maximizing QoE}
This section provides a comprehensive analysis of the ergodic covert rate performance in SSACC systems. Building upon the examination of fixed power allocation as discussed in (\ref{Average RB}), we extend our investigation to adaptive power allocation strategies. By optimizing power allocation, the objective is to enhance the CC rate under varying channel conditions. To achieve this, an optimal power allocation solution is derived using a GDM-based algorithm, designed to dynamically adjust to channel variations and maximize security performance.

\subsection{Problem Formulation}
To enhance system performance, our objective is to maximize QoE across various distance scenarios by optimizing the power allocation between Alice’s transmit power ($P_A$) and the upper bound of the jamming power ($\hat{P}_J$). Consequently, the optimization problem can be formulated as follows:
\begin{subequations}
    \label{optimal 1}
    \begin{alignat}{2}
        \underset{P_A,\hat{P}_J}{\text{maximize}}\ \ &QoE = \alpha \mathbf{E}(R_S) + \varsigma \mathbf{E}(R_W) \mathbf{E}(\mathbb{P}_{MDE}), \label{optimal target}\\ 
        \text{subject to}\ \ &P_A + \hat{P}_J\leq P_{max},\label{optimal power}\\
        &\mathbf{E}(\mathbb{P}_{MDE})\geq 1 - \Lambda ,\label{optimal AMDEP}\\
        &0\leq \alpha \leq 1, \label{optimal alpha}\\
        &D_{LRW}\leq d_{RW} \leq D_{HRW}, \label{optimal drw}\\
        &D_{LAR}\leq d_{AR} \leq D_{HAR}, \label{optimal dar}\\
        &D_{LRB}\leq d_{JR} \leq D_{HRB}, \label{optimal djr}\\
        &D_{LRB}\leq d_{RB} \leq D_{HRB}, \label{optimal drb}
    \end{alignat}
\end{subequations}
where $P_{max}$ denotes the total power constraint and $D_{\chi}$ denotes distance boundary constraints. The inequalities in (\ref{optimal drw}), (\ref{optimal dar}), (\ref{optimal djr}) and (\ref{optimal drb}) characterize the dynamic nature of the communication environment induced by user mobility. Specifically, unlike static scenarios, the locations of the mobile nodes are time-varying, causing the distances $d_{\chi}$ ($\chi \in \{AR, RW, JR, RB\}$) to fluctuate stochastically within the bounded intervals $[D_{L\chi}, D_{H\chi}]$. This introduces uncertainty into the large-scale fading, necessitating an adaptive power allocation strategy.

Specifically, the objective in (\ref{optimal target}) is to maximize QoE, as defined in (\ref{QoE}), while (\ref{optimal AMDEP}) represents the covert constraint, with $\Lambda$ as a predefined parameter specifying a target level of covertness, and $\mathbb{P}_{MDE}$ is given in (\ref{AMDEP Result}). Since secrecy capacity increases with $P_A$, the maximum value of the objective function (\ref{optimal target}) is achieved when the power constraint in (\ref{optimal power}) is satisfied with equality, i.e., $P_A + \hat{P}_J = P_{max}$. Consequently, we can simplify (\ref{optimal 1}) as follows:
\begin{subequations}
    \label{optimal 2}
    \begin{alignat}{2}
        \underset{\kappa}{\text{maximize}}\ \ &QoE,\label{optimal 2 target}\\ 
        \text{subject to}\ \ &0\leq \kappa \leq 1,\label{optimal 2 power}\\
        &\mathbf{E}(\mathbb{P}_{MDE})\geq 1 - \Lambda ,\label{optimal 2 AMDEP}\\
        &D_{LRW}\leq d_{RW} \leq D_{HRW},\\
        &D_{LAR}\leq d_{AR} \leq D_{HAR},\\
        &D_{LRB}\leq d_{JR} \leq D_{HRB},\\
        &D_{LRB}\leq d_{RB} \leq D_{HRB},
    \end{alignat}
\end{subequations}
where $\kappa = P_A /P_{max} = 1 - \hat{P}_J/P_{max}$ is the power allocation coefficient.

\subsection{Generative Diffusion Models Optimization}
In environments with varying distances, the closed-form expressions governing average channel capacity of Bob and secrecy capacity are analytically involved, making it challenging to directly obtain an optimal solution under the given constraints. Furthermore, certain parameters necessary for the optimization process are not available a priori and must be obtained in real time. To solve for the best power distribution, we employ a GMD-based RL model. While learning-based approaches have been successfully applied to other security domains like intrusion detection \cite{9668958}, here we leverage this emerging approach for network optimization problems \cite{du2023beyond}. Given that $\kappa$ is continuous, we apply proximal policy optimization (PPO) as the reinforcement learning policy, a widely used method capable of handling continuous action spaces \cite{schulman2017proximal}.
\subsubsection{Designed State Space, Action, and Reward Function} To address problem (\ref{optimal 2}), the state space, action, and reward function for the GDM Optimization (GDMOPT) framework are defined as follows:

\textit{State:} At each step $t(t=1,2,...,T)$, the state $s_t$ is composed of the total power limit $P_{max}$ and the path losses $\{\mathcal{L}_1, \mathcal{L}_2, \mathcal{L}_3, \mathcal{L}_4\}$ at that step.

\textit{Action:} The action $t(a_t)$ produced by the PPO policy at step $t$ corresponds to the parameter $\kappa$.

\textit{Reward:} At the conclusion of each step, a reward is assigned according to the following function:
\begin{equation}
    r_t=\begin{cases}
        \mathbf{E}(\mathbb{P}_{MDE})-1,&\mathbf{E}(\mathbb{P}_{MDE}) < 1-\Lambda,\\
        QoE,&\mathbf{E}(\mathbb{P}_{MDE})\geq 1-\Lambda.\\
    \end{cases}
\end{equation}

\subsubsection{Basic Principles of GDMOPT}
A deterministic value function $V_{\pi}$ can be defined for every policy $\pi$ \cite{9434397}. The PPO algorithm iteratively refines the policy, enabling the agent to learn from environmental feedback via the value function without requiring prior knowledge. During this process, actions that yield higher rewards are increasingly favored, while those with lower rewards are gradually suppressed, leading to convergence on a policy that maximizes cumulative long-term reward.

In the GDMOPT algorithm, GDM serves as a replacement for the action network in traditional DRL algorithms, with actions produced through a denoising process \cite{wang2022diffusion}. GDM operate through a forward diffusion process, wherein an initial input is progressively corrupted by the addition of Gaussian noise across a series of steps. These noisy samples act as training targets for a denoising network, beginning with the original data sample $\mathbf{x}_0$. At each step $t$, gaussian noise of variance $\beta_t$ is added to $\mathbf{x}_{t-1}$, producing $\mathbf{x}_t$ according to the conditional distribution $q(\mathbf{x}_t|\mathbf{x}_{t-1})$. This transition is mathematically described as follows:
\begin{equation}
    q\left( \mathbf{x}_t|\mathbf{x}_{t-1} \right) =\mathcal{N}\left( \mathbf{x}_t|\boldsymbol{\mu }_t=\sqrt{1-\beta _t}\mathbf{x}_{t-1}, \mathbf{\Sigma}_t=\beta_t\mathbf{I}\right) ,
\end{equation}
where $q\left( \mathbf{x}_t|\mathbf{x}_{t-1} \right)$ represents a normal distribution. In the reverse diffusion process, as $T$ becomes large, $x_{T}$ approaches an isotropic Gaussian distribution \cite{ho2020denoising}. The reverse distribution $q\left( \mathbf{x}_{t-1}|\mathbf{x}_t \right)$ can be learned and approximated by a parameterized model $p{\theta}$, expressed as follows:
\begin{equation}
    p_{\theta}\left( \mathbf{x}_{t-1}|\mathbf{x}_t \right) =\mathcal{N}\left( \mathbf{x}_{t-1}|\boldsymbol{\mu }_{\theta}\left( \mathbf{x}_t,t \right) , \mathbf{\Sigma}_t\left(\mathbf{x}_t,t\right) \right). 
\end{equation}
The reverse generative process, which reconstructs the trajectory from the noisy latent variable $\mathbf{x}_T$ to the clean data $\mathbf{x}_0$, is defined by the following equation:
\begin{equation}
    p_{\theta}\left( \mathbf{x}_{0:T} \right) =p_{\theta}\left( \mathbf{x}_T \right) \prod\limits_{t=1}^T{p_{\theta}\left( \mathbf{x}_{t-1}|\mathbf{x}_t \right)}.
\end{equation}

\subsubsection{The GDMOPT-Based Algorithm for Power Allocation Optimization}
\begin{algorithm}[t]
\caption{Objective Function and Solution Space Definitions}
\label{alg:compute_objective}
\begin{algorithmic}[1]
    \REQUIRE ~\\
    Channel state information $\mathbf{H}$; Action vector $\mathbf{s}$; Total power budget $P_T$; Weights $\alpha,\varsigma$; Reliability threshold $\Lambda$.
    \ENSURE ~\\
    Reward value $r$.

    \STATE \textbf{Initialize:} power allocation vector $\mathbf{p}$ based on action $\mathbf{s}$:
    \begin{equation*}
        \mathbf{p} \leftarrow P_T \frac{\mathbf{s}}{\sum_{m=1}^{M} s_m}
    \end{equation*}

    \STATE Calculate performance metrics $\mathbf{E}(R_S)$, $\mathbf{E}(\mathbb{P}_{MDE})$ and $\mathbf{E}(R_W)$ based on $\mathbf{H}$ and $\mathbf{p}$.
    \IF{$\mathbf{E}(\mathbb{P}_{MDE}) < 1 - \Lambda$}
        \STATE Penalty for violating constraint.
        \STATE $r \leftarrow \mathbf{E}(\mathbb{P}_{MDE}) - 1$
    \ELSE
        \STATE Calculate weighted objective utility.
        \STATE $r \leftarrow QoE$
    \ENDIF

    \RETURN $r$

  \end{algorithmic}
\end{algorithm}

\begin{algorithm}[t]
  \caption{Environment State Initialization}
  \label{alg:env_def}
  \begin{algorithmic}[1]

    \REQUIRE ~\\Dimension $M$; Boundary limits $[L_{min}, L_{max}]$.
    \ENSURE ~\\State vector $\mathbf{s}_{\text{env}}$.
    \STATE \textbf{Initialize:} state vector $\mathbf{s}_{\text{env}}$ with random positions:
    \STATE $\mathbf{s}_{\text{env}} \leftarrow \text{Generate } M \text{ values from } \mathcal{U}(L_{min}, L_{max})$
    \RETURN $\mathbf{s}_{\text{env}}$

  \end{algorithmic}
\end{algorithm}
\begin{algorithm}[t]
    \caption{GDMOPT in Network Optimization}
    \label{alg:GDM}
    \begin{algorithmic}[1]
        \REQUIRE~\\
        \textbf{GDMOPT parameters:}\\
        Denoising step number $N$, exploration noise $\epsilon$;\\
        Learning rates of the solution evaluation network and solution generation network;\\
        Replay buffer $\mathcal{B}$ size, mini–batch size $B$.\\[0.2em]
        \textbf{Environment:}\\
        Training trajectories of the environment state $g$, power constraints and channel statistics.\\[0.4em]

        \ENSURE~\\
        Trained solution generation network $\epsilon_{\theta}$;\\
        for a given environment, optimal power allocation $p_0$.\\[0.6em]
        \STATE \textbf{Initialize:} Empty replay buffer $\mathcal{B}$; randomly initialized
               solution generation network $\epsilon_{\theta}$ and solution
               evaluation network $Q_{\nu}$; initialize environment state $g^{(0)}$.
        \STATE Set the exploration process $\mathcal{N}$ (e.g., Gaussian noise) for power allocation exploration.

        \FOR {each training step $j = 1,2,\ldots,J$}
        \STATE Observe the current environment $g^{(j)}$ using Algorithm~\ref{alg:env_def}.
        \STATE Generate Gaussian noise $p_N$ and obtain
               power allocation $p_0^{(j)}$ by denoising $p_N$ with
               $\epsilon_{\theta}$.
        \STATE Add exploration noise from $\mathcal{N}$ to $p_0^{(j)}$.
        \STATE Apply $p_0^{(j)}$ to the environment and obtain the reward
               $r^{(j)}\!\left(g^{(j)}, p_0^{(j)}\right)$ via Algorithm~\ref{alg:compute_objective}.
        \STATE Observe the next environment state $g^{(j+1)}$.
        \STATE Store the transition
               $\big(g^{(j)}, p_0^{(j)}, r^{(j)}, g^{(j+1)}\big)$ in the replay buffer $\mathcal{B}$.
        \STATE Sample a mini–batch from $\mathcal{B}$ and update $Q_{\nu}$
               according to (\ref{Q-value}).
        \STATE Update the solution generation network $\epsilon_{\theta}$ according to
               (\ref{solution generation network}).
        \STATE Optionally perform target network soft update if needed.
        \ENDFOR
        \STATE \textbf{Inference:} For a given environment state $g$, generate the
               optimal power allocation $p_0$ by denoising Gaussian noise using
               the trained $\epsilon_{\theta}$.
        \STATE \textbf{return} $p_0$
    \end{algorithmic}
\end{algorithm}

Building on the GDMOPT framework and the optimization problem presented in (\ref{optimal 2}), we develop an algorithm for power allocation aimed at maximizing QoE while adhering to specified constraints. The procedure for addressing this optimization challenge via diffusion models is structured as follows:

\textit{Solution Space Definition:} The initial phase of the optimization process involves defining the solution domain. The AI-derived solution denotes the power allocation scheme that leads to maximum QoE. This solution is produced by the GDMOPT via multiple iterations of denoising applied to Gaussian noise, as detailed in \cite{du2023beyond}.

\textit{Objective Function Definition:} Here, the diffusion model is trained with the objective of improving QoE, as shown in Algorithm \ref{alg:compute_objective}.

\textit{Dynamic Environment Definition:} The GDMOPT aims to produce an optimized power allocation strategy adapted to specific channel distance conditions. These distances represent the large-scale fading components of the channel; thus, we consider a general scenario in which each channel distance, denoted as $d_\chi$ $\left(\chi \in \{AR, JR, RB, RW\}\right)$, varies randomly within a specified range, such as $\left(20, 70\right)$, as outlined in Algorithm \ref{alg:env_def}.

\textit{Training and Inference:} To achieve an optimized distribution of transmit power, a conditional GDMOPT is employed, which iteratively denoises the initial distribution to arrive at an optimized solution. The resulting power allocation scheme for a given environment is denoted by $p$. The GDMOPT transforms environmental states into power allocation strategies and is referred to as the solution generation network and represented by $\epsilon_{\theta}\left(p|g\right)$, where $\theta$ represents the neural network parameters. A deterministic power allocation is generated with the purpose of maximizing the expected utility outlined in Algorithm \ref{alg:compute_objective}. This is achieved through a reverse denoising process within a conditional GDMOPT framework, where the ultimate sample generated in this reverse sequence serves as the selected power allocation strategy.

In our optimization problem, the absence of an expert dataset necessitates the introduction of a solution evaluation network, denoted as $Q_\upsilon$. This network assigns a $Q_\upsilon$ to each pair of environment and power allocation $(g, p)$, which represents the expected value of the objective function. The construction of the optimal solution generation network is achieved through
\begin{equation}
    \label{solution generation network}
    \text{arg\ }\min _{\boldsymbol{\epsilon }_{\theta}}\mathcal{L}_{\boldsymbol{\epsilon }}\left( \theta \right) =-\mathbf{E}_{\boldsymbol{p}_0\backsim \boldsymbol{\epsilon }_{\theta}}\left[ Q_{\upsilon}\left( \mathbf{g,p}_0 \right) \right] .
\end{equation}
To ensure accuracy, the solution evaluation network $Q_\upsilon$ is trained with the objective of minimizing the deviation between its predicted Q-values and the actual Q-values. Consequently, the optimization task for $Q_\upsilon$ is defined as
\begin{equation}
    \label{Q-value}
    \operatorname*{arg\,min}_{Q_{\upsilon}} \mathcal{L}_Q\left( \upsilon \right) = \mathbf{E}_{\mathbf{p}_0 \sim \pi_{\theta}} \left[ \left( r\left( \mathbf{g}, \mathbf{p}_0 \right) - Q_{\upsilon}\left( \mathbf{g}, \mathbf{p}_0 \right) \right)^2 \right].
\end{equation}
Here, $r$ denotes the objective function value obtained when the generated power allocation strategy $p_0$ is applied to the environment state $g$. The overall GDM procedure for secrecy-capacity maximization is presented in Algorithm~\ref{alg:GDM}~\cite{du2023beyond}.

Under the standard assumption that the computational cost of a forward/backward pass grows linearly with the number of network parameters \cite{du2023beyond,3524483}, the overall training complexity of the proposed GDM-PPO algorithm can be approximated as follows. Let $N_{\pi}$ and $N_V$ denote the numbers of parameters in the diffusion-based policy network and the value network, respectively; $Z$ the number of denoising steps in the diffusion policy; $N_{\text{it}}$ the number of PPO iterations; $N_{\text{tr}}$ the number of environmental time steps (samples) collected per iteration; and $K$ the number of gradient epochs per iteration. In vanilla PPO, the dominant per-iteration complexity is on the order of $\mathcal{O}\!\big(K N_{\text{tr}} (N_{\pi} + N_V)\big)$ \cite{schulman2017proximal,3524483}. In GDM-PPO, each policy evaluation and update must propagate through $Z$ denoising steps, which effectively multiplies all policy-related computation by $Z$, following the same pattern as diffusion-based SAC and DDPG in \cite{3524483}. Consequently, the total training complexity scales as
\begin{equation}
    \mathcal{O}\!\big(N_{\text{it}} K N_{\text{tr}} (Z N_{\pi} + N_V)\big),
\end{equation}
while the inference-time complexity of generating a trajectory of length $T$ is $\mathcal{O}(T Z N_{\pi})$, i.e., linear in both the number of diffusion steps and the size of policy network.

\section{Numerical Results}
\label{sec:Numerical Result}
In this section, we assess the system performance by conducting numerical analysis, while the accuracy of the analytical results is confirmed through Monte Carlo simulations. Table \ref{tab:Parameters Setting} provides the simulation parameters, which are selected based on established settings in the literature \cite{lv2021covert,10210075}. Additionally, the reverse diffusion mechanism enables the generation of new data, facilitating a comprehensive examination of the system’s performance under various conditions.
\begin{table*}[t]
    \centering   
    \caption{Parameters Setting}
    \label{tab:Parameters Setting}
    \begin{tabular}{|c|c|}
        \hline
        Bandwidth                                       & $B=1$ MHz                                               \\ \hline
        Number of RIS reflecting elements               & $N = 8$ and $128$                                      \\ \hline
        Amplitude-reflection coeficient of the RIS      & $\beta = 0.9$                                          \\ \hline
        Distances                                       & $d_{AR}=d_{JR}=100$ m, $d_{RB}=50$ m, $d_{RW}=20$ m         \\ \hline
        Path-loss exponents                             & $\alpha_{AR}=3,\alpha_{JR}=\alpha_{RB}=\alpha_{RW}=2.5$ \\ \hline
        Nakagami fading pararmeters                     & $m_{AR}=m_{JR}=m_{RB}=m_{RW}=3$                        \\ \hline
        Noise power                                     & $\sigma^2_B=\sigma^2_W=-80$ dBm                         \\ \hline
        Number of points for Gauss-Laguerre quadratures & $u_1=u_2=100$                                          \\ \hline
    \end{tabular}
\end{table*}

\subsection{AMDEP, Average Capacities, and QoE}
\begin{figure}[]
    \centering
    \subfloat[]{\label{PJ_10dBm_AMDEP}\includegraphics[width=0.5\textwidth]{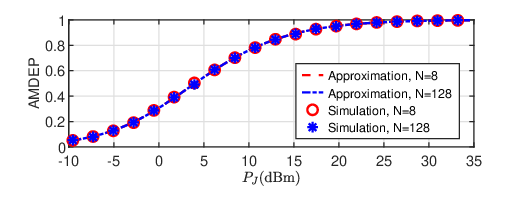}}\\
    \subfloat[]{\label{PA_10dBm_AMDEP}\includegraphics[width=0.5\textwidth]{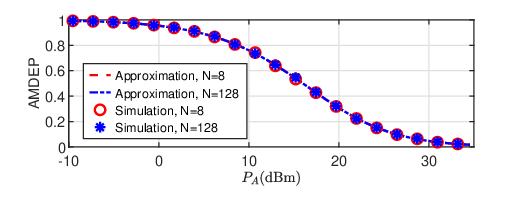}}
    \caption{AMDEP versus Alice’s transmit power and Jammer’s transmit power. $\left(a\right)\hat{P}_J=40$ dBm. $\left(b\right) P_A=40$ dBm.}\label{Pic_amdep}
\end{figure}

Fig. \ref{PJ_10dBm_AMDEP} illustrates the variation of Willie's AMDEP with respect to Alice’s transmit power, while Jammer's transmit power is set to 40 dBm. Similarly, Fig. \ref{PA_10dBm_AMDEP} illustrates the AMDEP of Willie with respect to the Jammer's transmit power, where Alice's transmit power is held constant at 40 dBm. For the two considered cases, with $N = 8$ and $N = 128$, the results show that the AMDEP decreases as $P_A$ increases, approaching 0 as $P_A \rightarrow \infty$ (as seen in Fig. \ref{PJ_10dBm_AMDEP}), and increases with $P_J$, converging to 1 as $\hat{P}_J \rightarrow \infty$ (as shown in Fig. \ref{PA_10dBm_AMDEP}). This behavior is consistent with expectations, as a larger jamming-to-signal power ratio increases the uncertainty faced by Willie. Additionally, the number of RIS elements has limited impact on the AMDEP, as the major factors influencing covert performance are the powers of the jammer and Alice. In both subfigures, the simulation results align closely with the analytical results provided by (\ref{AMDEP}), confirming thus the validity of our analysis. Furthermore, when the DEP threshold is higher than $80\%$, the derived analytical expressions are suitable for use in power allocation optimization.

\begin{figure}[]{\includegraphics[width=0.5\textwidth]{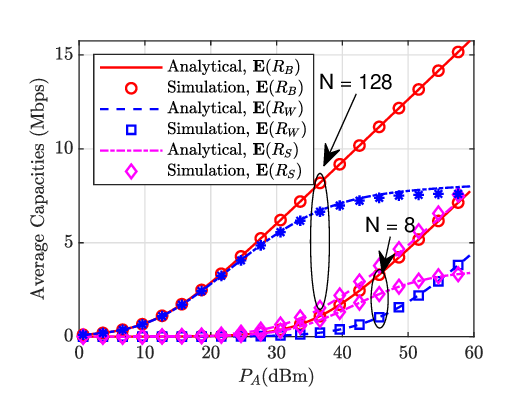}}
    \centering
    \caption{Average channel capacities and ASC versus Alice’s transmit power where $\hat{P}_J=40$ dBm.}\label{Pic_rate}
\end{figure}

Fig. \ref{Pic_rate} illustrates the relationship between Alice's transmit power and the average channel capacity of Bob, ASC, and the average channel capacity of Willie for $N=8$ and $N=128$. In particular, AMDEP is independent of $N$, attributed to the scale invariance where the channel scale parameter cancels out in the ratio of received powers. Initially, it is evident that the simulated average capacities and ASC are closely aligned with the analytical results obtained from (\ref{ASC}), (\ref{Average RB}), and (\ref{Average RW}). Furthermore, the average channel capacity of Bob increases with higher values of $P_A$. Additionally, an increase in $N$ leads to an enhancement of the ASC. The accuracy of this analysis supports the application of analytical results in the optimization of power allocation.

\begin{figure}[]{\includegraphics[width=0.5\textwidth]{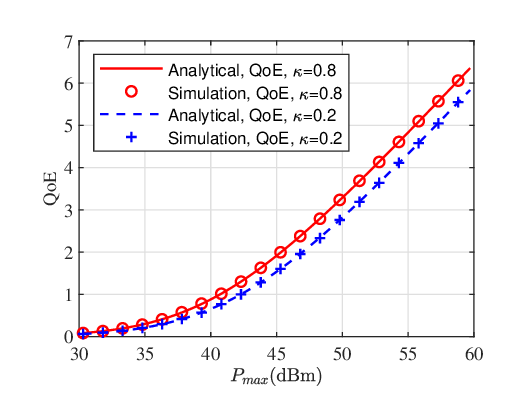}}
    \centering
    \caption{QoE versus $P_{max}$, where $\kappa = 0.8$ and $\kappa = 0.2$, respectively.}\label{Pic_qoe_05}
\end{figure}

As illustrated in Fig. \ref{Pic_qoe_05}, under fixed power distribution coefficients ($\kappa = 0.8$ and $\kappa = 0.2$) and with $N = 8$, QoE increases with rising $P_{max}$. It is observed that the AMDEP remains unaffected by $P_{max}$, as the minimum DEP in (\ref{DEP_optimal}) is contingent upon the ratio $P_A/\hat{P}_J$. Consequently, the expected AMDEP in (\ref{AMDEP Result}) depends on this ratio rather than being influenced by $P_{max}$. In contrast, it is noted that ASC increases with an increase in $P_{max}$.

\subsection{Power Allocation Optimization for Maximizing QoE}

\begin{figure}[]
    \centering
    \subfloat[]{\label{k_dep}\includegraphics[width=0.5\textwidth]{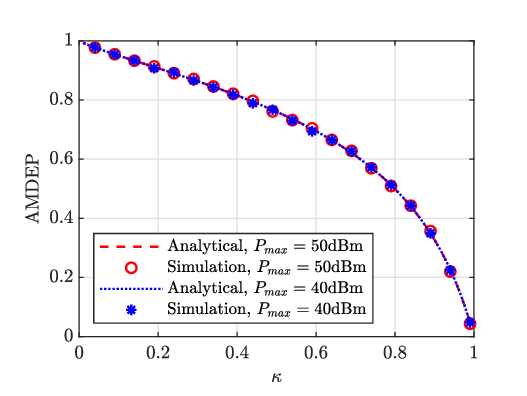}}\\
    \subfloat[]{\label{k_rate}\includegraphics[width=0.5\textwidth]{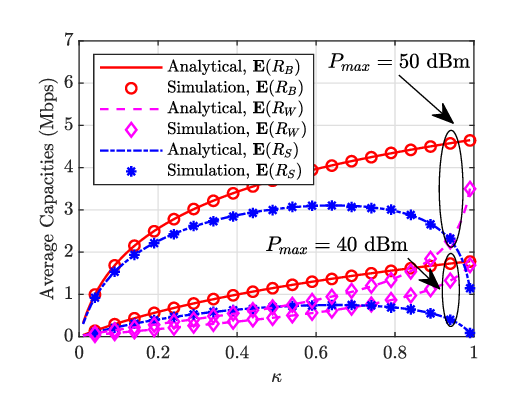}}\\
    \subfloat[]{\label{k_qoe}\includegraphics[width=0.5\textwidth]{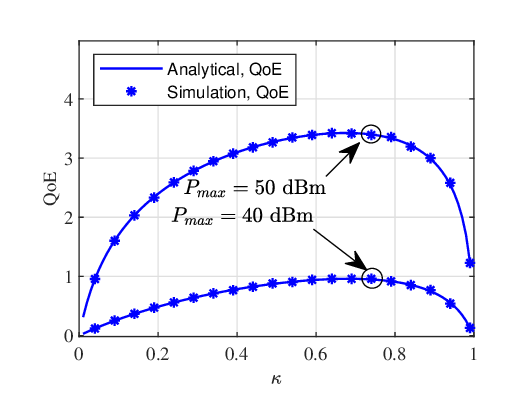}}
    \caption{AMDEP, ASC and QoE versus $\kappa$, where $P_{max}=40$ dBm and $P_{max}=50$ dBm, respectively.}\label{Pic_VersusKappa}
\end{figure}

We depict the variations of AMDEP as functions of $\kappa$ in Fig. \ref{k_dep}. As $\kappa$ increases, a decreasing trend is observed in AMDEP is observed. Since the AMDEP threshold is typically above $80\%$, the analytical results for AMDEP are suitable for use in the optimization problem. 

In Fig. \ref{k_rate}, it can be seen that as $\kappa$ increases, the ASC first grows, attaining a maximum, and then gradually decreases with further increases in $\kappa$ and the derived results align closely with the simulation outcomes. Therefore, the analytical results for ASC are also applicable to the optimization problem.

Fig. \ref{k_qoe} illustrates the variation of Quality of Experience (QoE) in relation to $\kappa$ (where $\kappa = P_A / P_{max} = 1 - \hat{P}_J / P_{max}$). The observed trend is primarily influenced by the ASC, exhibiting an initial increase followed by a subsequent decrease as $\kappa$ varies. In the second part of QoE, however, a declining AMDEP combined with a rapidly increasing $R_{W}$ results in a diminished impact on QoE. This leads to an overall trend that closely aligns with variations in ASC.

\begin{figure}[]{\includegraphics[width=0.5\textwidth]{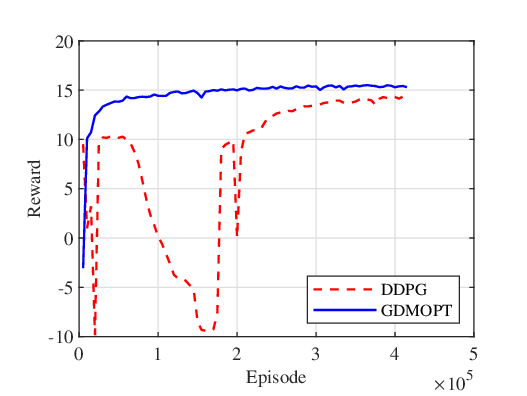}}
    \centering
    \caption{Test reward curves of GDM-aided and DDPG-aided optimization methods.}\label{GDM_Reward}
\end{figure}

Fig. \ref{GDM_Reward} illustrates the training process of the power allocation algorithm based on GDMOPT. It is observed that the total reward for the GDMOPT converges after approximately $100,000$ episodes. The final power allocation coefficient is represented by the converged action, while the resulting reward indicates the QoE achieved within the given constraints. The results reveal the conventional DRL method, i.e., DDPG, converges slowly, and its converged value is smaller compared to the GDM-based algorithm.

\begin{figure}[]{\includegraphics[width=0.5\textwidth]{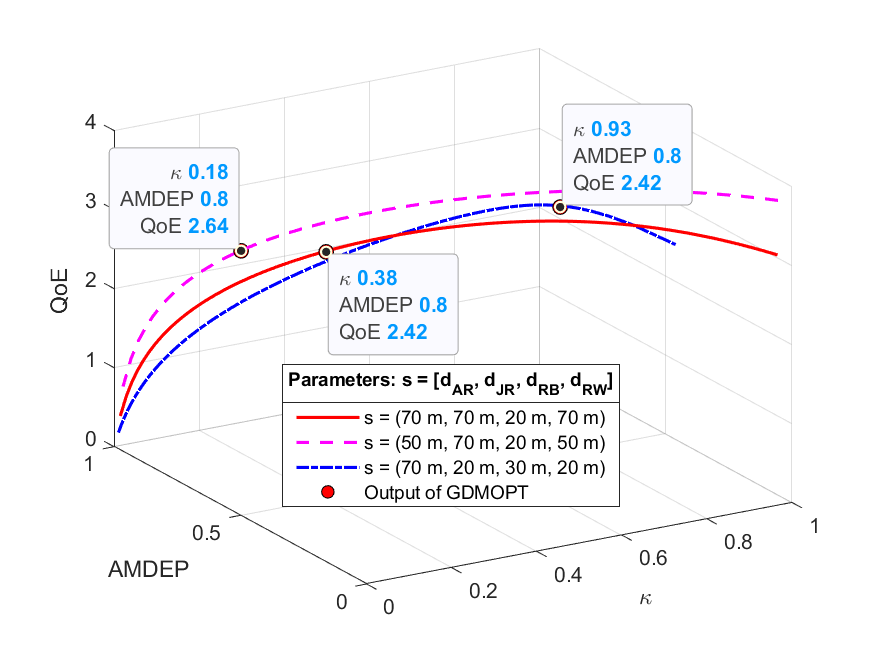}}
    \centering
    \caption{QoE and AMDEP versus $\kappa$ under different environmental distance parameters when $P_{\max} = 40$ dBm. The discrete markers indicate the optimal power allocation solutions output by the GDMOPT algorithm.}\label{RLvsMAT}
\end{figure}

To assess the effectiveness of the proposed GDM-based power allocation algorithm in mobile scenarios, the QoE and AMDEP trajectories are plotted against $\kappa$ for distinct environmental distance configurations in Fig. \ref{RLvsMAT}. Specifically, we define the objective of the maximization problem as zero if the constraints are not satisfied, and the output of GDMOPT matches the maximum value subject to the constraints.

\section{Conclusion}
\label{sec:Conclusion}
In this paper, we have proposed an SSACC system  that simultaneously ensures ASC and AMDEP to enhance transmission QoE. By deriving analytical expressions for both covertness and secrecy performance and designing a GDM-enhanced DRL algorithm, we demonstrated that the proposed approach can effectively optimize power allocation to improve system security under diverse user location conditions. The simulation results verified the advantages of the scheme in balancing covertness and secrecy capacity. This study underscores the potential of integrating RIS and DRL for SC and CC. Future research could explore more complex multi-user scenarios to further enhance transmission security and expand the applicability of the proposed framework.

\bibliographystyle{IEEEtran}
\bibliography{IEEEabrv, reference}

\vspace{11pt}
\vfill
\end{document}